\documentclass[letterpaper,10pt]{article}

\usepackage{amsthm,amsmath,amssymb,amsfonts}
\usepackage{mathpazo}
\usepackage{natbib}
\def\cite{\citep}

\usepackage[T1]{fontenc}
\usepackage[utf8]{inputenc}

\usepackage{thmtools,thm-restate}

\usepackage{etoolbox}
\newtoggle{short}

\newif\ifshort
\shorttrue

\newcommand{\shortsum}{\iftoggle{short}{\textstyle\sum}{\sum}}

\newtheorem{thm}{Theorem}[section]

\newtheorem{cor}[thm]{Corollary}
\newtheorem{lem}[thm]{Lemma}
\newtheorem{ex}[thm]{Example}
\newtheorem{fact}[thm]{Fact}
\usepackage{mathtools} %
\usepackage{tabularx}
\usepackage[vlined,ruled,boxed,linesnumbered]{algorithm2e}
\DontPrintSemicolon
\SetKw{True}{True}
\SetKw{False}{False}
\SetKw{Break}{break}
\SetKw{continue}{continue}
\SetKwFor{RepTimes}{repeat}{times}{end repeat}
\SetKwFor{Loop}{Loop}{}{end loop}

\usepackage{enumitem}

\newcommand{\setplus}{\ensuremath{\oplus}}

\newcommand{\vanp}{vanish-prime}
\newcommand{\diffp}{difference-prime}

\newcommand{\la}{\langle}
\newcommand{\ra}{\rangle}

\newcommand{\bigO}{\ensuremath{\mathcal{O}}}
\newcommand{\bigoh}{\bigO}
\newcommand{\softohmu}{\ensuremath{\widetilde{\mathcal{O}_\mu}}}
\newcommand{\softoh}{\ensuremath{\widetilde{\mathcal{O}}}}

\newcommand{\ring}{\ensuremath{\mathsf{R}}}
\newcommand{\R}{\ring}

\newcommand{\FF}{\ensuremath{\mathbb{F}}}
\newcommand{\RR}{\ensuremath{\mathbb{R}}}

\newcommand{\ZZ}{\ensuremath{\mathbb{Z}}}
\newcommand{\QQ}{\ensuremath{\mathbb{Q}}}

\newcommand{\fail}{\ensuremath{\mathtt{Fail}}}

\renewcommand{\S}{\ensuremath{\mathbb{S}}}

\newcommand{\V}{\ensuremath{V}}

\newcommand{\rem}{\ensuremath{\ \mathrm{rem}}\ }
\newcommand{\brem}{\ensuremath{\ \overline{\mathrm{rem}}\ }}
\newcommand{\polylog}{\ensuremath{\mathsf{polylog}}}

\DeclareMathOperator{\poss}{\ensuremath{\mathsf{poss}}}
\DeclareMathOperator{\supp}{\ensuremath{\mathsf{supp}}}

\newcommand{\norm}[1]{\ensuremath{\left\lVert#1\right\rVert}}
\newcommand{\A}{\mathcal{A}}
\newcommand{\B}{\mathcal{B}}

\DeclareMathOperator{\width}{\ensuremath{\mathsf{diam}}}

\newcommand*{\MyDef}{\mathrm{def}}
\newcommand*{\eqdefU}{\ensuremath{\mathop{\overset{\MyDef}{=}}}}%
\newcommand*{\eqdef}{\,\mathop{\overset{\MyDef}{\resizebox{\widthof{\eqdefU}}{\heightof{=}}{=}}}\,}

\usepackage[svgnames]{xcolor}
\definecolor{darkgreen}{rgb}{0,.35,0}
\definecolor{darkblue}{rgb}{0,0,.35}
\definecolor{darkred}{rgb}{.35,0,0}

\settoggle{short}{false}

\newcommand{\tit}{Output-sensitive algorithms for
sumset and sparse polynomial multiplication}

\usepackage[pdfauthor={Andrew Arnold and Daniel S. Roche},
  pdftitle={\tit},
  colorlinks, linkcolor=DarkBlue, 
  citecolor=DarkGreen,urlcolor=darkred, bookmarksnumbered,
]{hyperref}

\usepackage{doi}

\author{Andrew Arnold \and Daniel S. Roche}

\title{\tit}
\author{Andrew Arnold\\
\small Cheriton School of Computer Science\\
\small University of Waterloo\\
\small Waterloo, Ontario, Canada\\
\href{http://www.AndrewArnold.ca}{\tt a4arnold@uwaterloo.ca}
\and
Daniel S.\ Roche\\
\small Computer Science Department\\
\small United States Naval Academy\\
\small Annapolis, Maryland, USA\\
\href{http://www.usna.edu/cs/roche/}{\tt roche@usna.edu}}

\begin{document}

\maketitle

\begin{abstract}

We present randomized algorithms to compute the sumset  
(Minkowski sum) of two integer sets, and to multiply two
univariate integer polynomials given by sparse representations.
Our algorithm for sumset has cost softly linear in the combined size of the inputs
and output.  This is used as part of our sparse multiplication algorithm,
whose cost is softly linear in the combined size of the inputs,
output, and the sumset of the supports of the inputs.  
As a subroutine, we present a new method for computing the coefficients
of a sparse polynomial, given a set containing its support.
Our multiplication algorithm extends to multivariate Laurent
polynomials over finite fields and rational numbers.  Our
techniques are based on sparse interpolation algorithms and
results from analytic number theory.

\end{abstract}

\section{Introduction}

Sparse polynomials are a fundamental object in computer algebra.
Computer algebra programs including 
Maple, Mathematica, Sage, and Singular use a sparse
representation by default for multivariate polynomials, and there has
been considerable recent work on how to efficiently store and compute
with sparse polynomials \cite{Fat03,GL06,MP09,Roc11,HL12a}.

However, despite the memory advantage of sparse polynomials,
the alternative dense representation is still widely used for an
obvious reason: speed.
It is now classical \cite{CanKal91}
that two degree-$D$ dense polynomials can be multiplied in softly linear
time: $\bigoh(D\log D\log\log D)$ ring operations, and even better in
many cases \cite{HarvdHLec14b}. By contrast, two size-$T$ sparse
polynomials require $\bigoh(T^2)$ operations, and this
excludes the potentially significant cost of exponent
arithmetic.

Much of the recent work on sparse arithmetic has focused on ``somewhat
dense'' or structured cases, where the sparsity of the product is
sub-quadratic
\cite{MP09,Roc11,HL12a}. At the same time, sparse interpolation
algorithms, which in the fastest case can learn an unknown $T$-sparse
polynomial from $\bigoh(T)$ evaluations, have gained renewed interest
\cite{CL08,JM10,CKP12,ArnGieRoc14b}. 

Most closely related to the current work, \cite{vdHLec13}
recently presented algorithms to discover the coefficients of a
sparse polynomial product, provided a list of the exponents and some
preprocessing.  
In the context of pattern matching problems,
\cite{CH02} gave a Las Vegas algorithm to multiply
sparse polynomials with nonnegative integer coefficients
whose cost is $\softoh(T \log^2 D)$.

A remaining question is
whether output-sensitive sparse multiplication is possible in
time comparable to that of dense
multiplication. This paper answers that question,
with three provisos: First, our complexity is proportional to the
``structural sparsity'' of the output that accounts for exponent
collisions but not coefficient cancellations; 
second, our algorithms are randomized and may
produce incorrect results with controllably low probability;
and third, we ignore logarithmic factors in the size
of the input.

To explain the first proviso, define for a polynomial $F$ its
\emph{support} $\supp(F)$
to be the set of exponents of nonzero terms in $F$. The \emph{sparsity} of $F$,
written $\#F$, is exactly $\#\supp(F)$. 
For two polyomials $F$ and $G$, we have
$\#\supp(FG) \le \#F \cdot \#G$.
But in
many cases the set of \emph{possible exponents}
$$\poss(F,G) \eqdef \{e_F+e_G : e_F\in
\supp(F), e_G\in\supp(G)\}$$
is much smaller than $\#F \cdot \#G$. This \emph{structural sparsity}
$T=\#\poss(F,G)$, is an upper bound on the
actual sparsity $S=\#\supp(FG)$ of the product.
Strict inequality $S < T$ occurs only
in the presence of \emph{coefficient
cancellations}. Part of our algorithm's cost
depends only on the actual sparsity, and part depends
on the potentially-larger structural sparsity.

Our algorithms have not yet been carefully implemented, and we do
not claim that they would be faster than the excellent
software of \cite{GL06,MP13} and others for a wide range of practical
problems. However, this complexity improvement indicates that the
barriers between sparse and dense arithmetic may be weaker than we once
thought, and we hope our work will lead to practical improvements in the
near future.

\subsection{Our contributions}

Our main algorithm is summarized in Theorem~\ref{thm:main}.
Here and throughout, we rely on a version of ``soft-oh'' notation
that also accounts for a bound $\mu$ on the probability of failure: 
$\softohmu(\phi) \eqdef \bigoh\big(\phi \cdot\polylog(\phi/\mu)\big)$, for
any function $\phi$, where $\polylog$ means $\log^c$ for some fixed
$c>0$ \cite[see sec.~25.7]{MCA}.

\begin{thm}\label{thm:main}
  Given $F,G\in\ZZ[x]$ with degree bound $D > \deg F + \deg G$ and height
  bound
  $C \ge \norm{F}_\infty + \norm{G}_\infty$, and $\mu\in(0,1)$,
  Algorithm 
  \ref{proc:sparsemulzz} correctly computes the product $H=FG$
  with probability exceeding $1-\mu$,
  using worst-case expected $\softohmu(S\log C + T\log D)$ bit operations,
  where $S = \#\supp(FG)$ and $T = \#\poss(F,G)$ are the actual and
  structural sparsity of the product, respectively.
\end{thm}

\iftoggle{short}{}{
The linear dependence on $\log D$ and $\log C$ means this
result extends easily to multivariate polynomials and finite field
coefficients, using the Kronecker substitution.
}

Our algorithm relies on two subroutines, both of
which are based on techniques from sparse interpolation and rely on
number-theoretic results on the availability of primes.

The first subroutine \ref{proc:sumset}$(\A,\B)$ computes the
\emph{sumset} of two sets of integers $\A$ and $\B$, defined as
$$\A \setplus \B \eqdef \{a+b : a\in \A, b \in \B\}.$$
This algorithm, which may be of independent interest, has softly-linear
complexity in the size of the output $\A \setplus \B$.

The second subroutine \ref{proc:knownsupp}($F,G,\S$) 
requires a set containing $\supp(FG)$ in order to compute $FG$
in time softly-linear in the input and output sizes.
It is based on an algorithm in \cite{vdHLec13}, but is more efficient
for large exponents.

The main steps of our multiplication algorithm are:

\begin{enumerate}[noitemsep,nolistsep]
  \item Use \ref{proc:sumset} to compute $\poss(F,G)$.
  \item Run \ref{proc:knownsupp} with $\S=\poss(F,G)$ but with
    smaller coefficients, to discover the true $\supp(FG)$.
  \item Run \ref{proc:knownsupp} again, with the smaller exponent
    set $\supp(FG)$ but with the full coefficients.
\end{enumerate}

Steps 1 and 2 work with a size-$T$ exponent set but with small
coefficients, and both contribute $\softohmu(T\log D)$ to the overall bit
complexity. Step 3 uses the size-$S$ true support but with the full
coefficients, and requires $\softohmu(S (\log D + \log C))$ bit
operations, for a total of $\softohmu(T\log D + S\log C)$.

\subsection{Organization of the paper}

Section~\ref{sec:prelim} states our notational conventions and some
standard results, and Section~\ref{sec:redexp} contains the technical 
number theoretic results on which we base our
randomizations.

Section~\ref{sec:basecase} revisits and adapts our sparse interpolation
algorithm from ISSAC 2014 that will be a subroutine for our sumset algorithm,
presented in section~\ref{sec:sumset}.

Our new method to find the coefficients, once the support is known, is
presented in Section~\ref{sec:knownsupp}.
This is then used in concert
with our sumset algorithm in Section~\ref{sec:smul} to describe fully the
algorithm of Theorem~\ref{thm:main}, and also to explain how this can be
easily extended to output-sensitive sparse multiplication over 
$\R[x_1^{\pm 1},\ldots,x_n^{\pm 1}]$, where $\R$ is $\ZZ_m$, $\QQ$, or $\mathsf{GF}(p^e)$.

\section{Background and Preliminaries}\label{sec:prelim}

We count the cost of our algorithms in terms of bit complexity on a
random-access machine.  We state their costs using $\softohmu$ notation,
meaning that our algorithms have a factor
$\log^c \tfrac{1}{\mu}$ in the running time. We can make $c=1$
by running the entire algorithm with error bound $\tfrac{2}{3}$ some
$\bigoh(\log\tfrac{1}{\mu})$ times, then returning the most frequent
result.

Our main algorithm depends on an unknown number-theoretic constant, as
discussed in Section \ref{sec:redexp}. Thus we have proven only the
\emph{existence} of a Monte Carlo algorithm. We also discuss how this
could be easily handled in practice.

Our randomized procedures return either the correct answer
(with controllable probability $1-\mu$), or an incorrect
answer, or the symbol \fail. Whenever a subroutine returns
\fail{}, we assume the calling procedure returns \fail{} as well.

\subsection{Notation and Representations}\label{sec:assumptions}

We let $\ring$ denote a commutative ring with identity.  For $F \in \ring[x]$ we let $\langle F \rangle \subset \ring[x]$ denote the ideal generated by $F$.  For $n, m
\in \ZZ$, $m > 0$, we let $n \brem m$ and $n \rem m$ denote the 
integers $s \in [0,m)$ and  $t \in [-m/2,m/2)$ respectively,
such that $n \equiv s \equiv t \pmod{m}$.  
We write $\ZZ_m$ for $\ZZ/m\ZZ$, typically represented as
$\{n \rem m | n\in\ZZ\}$.

Unless otherwise stated we assume $F\in\R[x]$ is of the form
\begin{equation}\label{eqn:sparserep}
\textstyle F(x) = \sum_{1\le i\le S} c_ix^{e_i},
\end{equation}
with coefficients $c_i\in\R$ and exponents
$e_i\in\ZZ_{\ge 0}$.  Often we assume each $c_i\ne 0$ and the
exponents are sorted $e_1<\cdots<e_S$, but it is sometimes useful
to relax these conditions.

We write $S\ge\#F$ and $D>\deg F$ for the sparsity
and degree bounds. When $\ring=\ZZ$ we write
$C>\norm{F}_\infty \eqdef \max_i |c_i|$ for the {\em height}
of $F$.  We also use the norm $\norm{F}_1 = \sum_i |c_i|$.

More generally, we consider multivariate Laurent polynomials
\iftoggle{short}{%
$F = \sum_{i=1}^S c_ix_1^{e_{i1}}\cdots x_n^{e_{in}} \in \ring[x_1^{\pm
1}, \dots, x_n^{\pm 1}]$. 
}{%
$$F = \sum_{i=1}^S c_ix_1^{e_{i1}}\cdots x_n^{e_{in}} \in \ring[x_1^{\pm
1}, \dots, x_n^{\pm 1}].$$
}%
In the case $\ring=\ZZ$, the \emph{sparse representation} of $F$
consists of a tuple
$(n,C,D,S)$, followed by a list of $S$ tuples 
$(c_i, (e_{i1},\ldots,e_{in}))$, where each $c_i$ is stored using
$\Theta(\log C)$ bits and each $e_{ij}$ is stored using $\Theta(\log D)$
bits. 
\iftoggle{short}{The total size in this case is
$\Theta(S\log C + Sn\log D)$. }{}%
When $\ring$ is instead a finite ring, $C$ is omitted and each
 $c_i$ is stored using $\Theta(\log|\ring|)$ bits.

When multiplying $F,G\in\ZZ[x]$, we assume shared bounds $C$ and $D$
so that total input/output size is
\begin{equation}\label{eqn:sparsesize}
\softoh((\#F + \#G + \#(FG))(\log C + \log D)),
\end{equation}
given that $\norm{FG}_1 < C^2\min(\#F,\#G)$.

The \emph{dense representation} of an $n$-variate
polynomial $F$ 
is an $n$-dimensional array of $D^n$
coefficients, where exponents are implicitly stored as array
indices. In the case that $\ring=\ZZ$ with bound $C$ as above, this
requires $\Theta(D^n \log C)$ bits.

The terms ``sparse polynomial'' and ``dense polynomial'' refer 
only to the choice of representation, and not to the relative 
number of nonzero coefficients. Typically, we assume $F$ is 
sparse and reserve $\tilde{F}$ to indicate a dense polynomial.
Converting between the sparse and dense representations 
has softly linear cost in the combined
input/output size.

When computing a sumset $A \oplus B$,
we assume that every integer in $A$ or $B$ is
represented using $\Theta(\log D)$
bits, where $D > \max(\norm{\A}_\infty,\norm{B}_\infty)$.
\iftoggle{short}{}{In this setting the combined bit-size of $A,B$, and $A \oplus B$
is $\softoh( (\#A + \#B + \#(A \oplus B))\log D)$.}

\subsection{Integer and Polynomial Arithmetic}

We cite the following results from integer and polynomial
arithmetic, which we use throughout.

\iftoggle{short}{
\begin{fact} \cite[Ch.~8, 10]{MCA} \label{fact:int_arithmetic} \label{fact:dense_arithmetic} \label{fact:crt}
  The following can be computed in softly-linear time in the bit-length
  of the inputs:
  \begin{itemize}[noitemsep,nolistsep]
    \item $m\pm n$, $mn$, $m\rem n$, $m \brem n$, for any $m,n\in\ZZ$ 
    \item Arithmetic in $\ZZ_m$ for any $m\in\ZZ_{>0}$ 
    \item Arithmetic on dense polynomials in $\ZZ_m[x]$
    \item CRT: given integers
    $(v_i,m_i),\ldots,(v_N,m_N)$,
    determine $v\in\ZZ$, $v < \prod_i m_i$, such that
      $\forall i, v \equiv v_i \pmod{m_i}$.
  \end{itemize}
\end{fact}

}{ %
\begin{fact}[\cite{HarvdHLec14}, Thm~9.8; \cite{MCA}, Cor.~9.9]
  \label{fact:int_arithmetic}
  Let $m,n \in \ZZ$.  Then the following may be computed using
  $\softoh( \log m + \log n)$ operations: $m \pm n, mn, m \rem n,
  m \brem n$.  Also, arithmetic operations in $\ZZ_n$
  require $\softoh( \log n)$ bit operations.
\end{fact}

\begin{fact}[\cite{MCA}, Thm.~10.25]\label{fact:crt}
  Given $m_i \in \ZZ_{>0}$, and $v_i \brem m_i$ for $1 \leq i \leq t$
  and $M = \prod_{i=1}^t m_i$, one can compute the solution $v \in
  [0,M)$ to the set of congruences $v \equiv v_i \pmod{m_i}$, $1 \leq i
  \leq t$ using $\softoh( \log M)$ bit operations.
\end{fact}

We use Chinese Remaindering to construct exponents from sets of congruences.  We also use dense polynomial arithmetic as a subroutine, and cite the following results.
\begin{fact}[\cite{HarvdHLec14b}]\label{fact:dense_arithmetic}
  Let $F,G\in\ZZ_m[x]$, $\deg F, \deg G < D$.
  Then the following may be computed using $\softoh(D\log m)$ bit
  operations: $F\pm G$, $FG$, $F \rem G$.
\end{fact}

In particular, dense arithmetic operations in $\ZZ_m[x]/\la x^p-1 \ra$ may be computed 
in $\softoh( p\log m)$ bit operations.
}

Assume $F\in\ZZ[x]$ as in \eqref{eqn:sparserep} with bounds
$D$, $S$, and $C$ as described above.
Our algorithm performs arithmetic on modular images of $F$.  For $F \in \ZZ_m[x]$, 
we represent $F(x) \bmod (x^p-1) \in \ZZ_m[x]/\la x^p-1 \ra$ by the remainder from dividing $\ZZ_m(x)$ by 
$(x^p-1)$.
\iftoggle{short}{}{Namely $F(x)\bmod(x^p-1)$ is a polynomial with degree less
than $p$. }%
Note we treat $F(x) \rem (x^p-1)$ and $F(x) \bmod 
(x^p-1)$ as elements of $\ring[x]$ and $\ring[x]/\la 
x^p-1\ra$ respectively.  To reduce a sparse polynomial $F 
\bmod (x^p-1)$, we reduce each exponent modulo $p$, and then add like-degree
terms.  By Fact \ref{fact:int_arithmetic}, we have: 
\begin{cor}\label{cor:modxp_m}
	Given any $F \in \ZZ_m[x]$, we can compute $F \bmod (x^p-1)$ using 
	$\softoh( S(\log D + \log m))$ bit operations.
\end{cor}

\section{Number-theoretic subroutines}\label{sec:redexp}

\subsection{Choosing primes}

We first recall how to choose a random prime number.

\begin{fact}[Corollary 3, \cite{RosSch62}]\label{fact:rosser}
If $\lambda \ge 21$, then
the number of primes in $(\lambda, 2\lambda]$  is at least $3\lambda/(5\ln \lambda)$.
\end{fact}
We test if $p$ is prime in $\bigoh( \polylog(p))$ time via the method in \cite{AgrKaySax04}. 
This test and Fact \ref{fact:rosser} lead to procedure \ref{proc:getprime}.

\begin{procedure}[Htbp]
  \caption{GetPrime($\lambda, \mu$)}\label{proc:getprime}
  \KwIn{$\lambda \geq 21$; $\mu \in (0,1)$.}
  \KwOut{Integer $p\in(\lambda,2\lambda]$,
  s.t.\ $\Pr[p\text{ not prime}] < \mu$.}
  \RepTimes{$m = \lceil (5/6) \ln \lambda \ln(1/\mu) \rceil$}{
    $p \gets$ random odd integer from $(\lambda,2\lambda] \cap \ZZ$ \;
    \lIf{$p$ is prime}{\Return $p$}
  }
  \Return \fail\;
\end{procedure}

\begin{lem}
  \label{lem:getprime_cost}
  \ref{proc:getprime}($\lambda, \mu$) works as stated and
  has bit complexity 
  $\softohmu(\polylog(\lambda)).$
\end{lem}
\begin{proof}
  The stated cost follows from fast primality testing due to \cite{AgrKaySax04}.
  The probability that any chosen $p$ is prime is at least
  $6/(5\ln\lambda)$, from Fact~\ref{fact:rosser}. 
  Therefore, using the fact that $(1-x)<\exp(-x)$ for any nonzero $x\in\RR$,
  the probability that none of the chosen $p$ are prime is at most
  $\left(1 - \tfrac{6}{5\ln\lambda}\right)^m 
    < \exp\left(\tfrac{-6m}{5\ln\lambda}\right) \le
    \mu,$ as desired.
  \end{proof}

It is frequently useful to choose a random prime that divides very few
of the integers in some unknown set $\S \subset \ZZ$. 
If a fraction of $\gamma$ integers in $\S$ do not vanish modulo 
$p$, then we call $p$ a \emph{$\gamma$-\vanp{}} for $\S$.  We 
call a $1$-\vanp{} for $\S$ a {\em good \vanp{}}.
Procedure \ref{proc:vanp} shows how to choose a random
$\gamma$-\vanp{}.

\begin{procedure}[tb]
  \caption{GetVanishPrime($S,D,\gamma,\mu$)\label{proc:vanp}}
  \KwIn{Integers $S,D\in\ZZ_{>0}$; $\gamma\in(0,1]$; $\mu\in(0,1)$.}
  \KwOut{Integer $p$, s.t.\ for
    any set $\S$ satisfying $\#\S\le S$ and $\norm{\S}_\infty<D$,
    with probability at least $1-\mu$, $p$ is a $\gamma$-\vanp{} for $\S$.}
  $\lambda \gets \displaystyle \max\left(21, 
    \tfrac{10}{3\mu}\min\left(S,\tfrac{1}{1-\gamma}\right)\ln D
    \right)$ \;
  \Return \ref{proc:getprime}($\lambda$, $\mu/2$)
\end{procedure}

\begin{lem}\label{lem:vanp}
  Procedure \ref{proc:vanp} works as stated to produce a
  $\gamma$-\vanp{} $p$ satisfying
  $$p\in\bigoh\left(\tfrac{1}{\mu}\min\left(S,\tfrac{1}{1-\gamma}\right) \log D
  \right)$$
  and has bit complexity $\polylog(p)$.
\end{lem}
\begin{proof}
  Let $\S$ be any subset of integers with $\#\S\le S$ and
  $\norm{S}_\infty < D$.
  Write $M = \prod_{a\in\S} |a| < D^S$, and write $k$ for the number of
  ``bad primes'' for which more than $(1-\gamma)S$ elements of $\S$
  vanish modulo $p$. Since each $p\ge\lambda$, this means that
  $\lambda^{(1-\gamma)Sk} \leq M$, and because $M<D^S$,
  $k < \ln D / ((1-\gamma)\ln \lambda)$ is an upper bound on the number
  of bad primes.

  If $1-\gamma$ is very small, we
  instead use a similar argument to say that the number of primes for
  which \emph{any} element of $\S$ vanishes is at most
  $k < S\ln D / \ln \lambda$. 

  Then Fact~\ref{fact:rosser} guarantees the prevalence of bad
  primes among all primes in $(\lambda,2\lambda)$ is at most $\mu/2$,
  so the probability of getting a bad prime, or of erroneously returning
  a composite $p$, is bounded by $\mu$.
\end{proof}

\subsection{Avoiding collisions}

A closely related problem is to choose $p$ so that most
integers in a set $\S$ are unique modulo $p$. We say that
$a\in\S$ \emph{collides} modulo $p$ if there exists $b\in\S$ with
$a\equiv b \pmod{p}$. We say $p$ is a \emph{$\gamma$-\diffp{} for
$\S$} if the fraction of integers in $\S$ which do not collide modulo $p$ is at
least $\gamma$. A 1-\diffp{} is called a \emph{good \diffp{} for \S}.
Procedure \ref{proc:diffp} shows how to compute \diffp{}s,
conditioned on the \emph{diameter} of the unknown set $\S$:
$$
\width(\S)\eqdef\max(\S)-\min(\S).
$$\vspace{-2em}
\begin{lem}\label{lem:diffp}
  Procedure \ref{proc:diffp} has bit complexity $\polylog(p)$ and works as stated to produce a
  $\gamma$-\diffp{} $p$ satisfying $p\in\bigoh(D)$ and
  $$p\in\bigoh\left(\tfrac{1}{\mu}S \min\left(S,\tfrac{1}{1-\gamma}\right) \log D
  \right).$$
\end{lem}
\begin{proof}
  Let $\S$ be any set as described. An element $a\in\S$ collides modulo
  $p$ iff the product of differences 
  $\prod_{b\in\S, b\ne a} (a-b)$ vanishes modulo $p$. If $p>D$
  this can never happen. Otherwise, as each such
  product is at most $\width(\S)^{S-1} < D^{S-1}$, the result follows
  from the Lemma~\ref{lem:vanp}, setting the $D$ of the lemma to
  $D^{S-1}$.
\end{proof}

\begin{procedure}[tb]
  \caption{GetDiffPrime($S,D,\gamma,\mu$)\label{proc:diffp}}
  \KwIn{Integers $S,D\in\ZZ_{>0}$; $\gamma\in(0,1]$; $\mu\in(0,1)$.}
  \KwOut{Integer $p$, s.t.\ for
    any set $\S$ satisfying $\#\S\le S$ and $\width(\S)<D$,
    with probability at least $1-\mu$, $p$ is a 
    $\gamma$-\diffp{} for $\S$.}
  $\lambda \gets \displaystyle
    \tfrac{10}{3\mu}(S-1) \min\left(S,\tfrac{1}{1-\gamma}\right)\ln D$
    \;
  \lIf{$\lambda < 21$}{\Return \ref{proc:getprime}$(21,\mu/2)$}
  \lElseIf{$\lambda > D$}{\Return \ref{proc:getprime}$(D,\mu/2)$}
  \lElse{\Return \ref{proc:getprime}$(\lambda,\mu/2)$}
\end{procedure}

Our algorithms often perform arithmetic modulo $(x^p-1)$.
Similar to the notion of collisions above for a set of integers modulo
$p$, we say two distinct terms $cx^e$ and $c'x^{e'}$ of $F\in\R[x]$
\emph{collide} modulo $(x^p-1)$ if $e \equiv e' \pmod{p}$.

\iftoggle{short}{}{
\begin{ex}
Let $F=x + x^{6}+3x^7$.  Then $F \bmod (x^5-1)=2x + 3x^2$.  The term $2x$ is the image of $x + x^6$.  We say $x$ and $x^6$ collide modulo $(x^5-1)$, whereas $3x^7$ uniquely maps to $3x^2$.
\end{ex}
}

Essentially, reduction modulo $(x^p-1)$ ``hashes'' exponent 
$e\in\supp(F)$ to $e \brem p$. If $p$ is a good \diffp{}
for $\supp(F)$ and $q$ is a good \vanp{} for the coefficients of
$F$, then $F \rem (x^p-1)$ with coefficients reduced modulo $q$ has 
the same sparsity as $F$ itself.

\subsection{Primes in arithmetic progressions}\label{ssec:pqw}

Sometimes we implicitly 
reduce exponents modulo $p$ by evaluating at
$p$th roots of unity. In such cases we need to construct primes $q$ such
that $p|(q-1)$, and to find $p$th roots of unity modulo each $q$. 

In principle, this procedure is no different than the previous
ones, as there is ample practical and theoretical evidence to suggest
that the prevalence of primes in arithmetic progressions without 
common divisors is roughly the
same as their prevalence over the integers in general.

However, the closest to Fact~\ref{fact:rosser} that we can get here is
as follows, which is a special case of Lemma~7 in \cite{Fou13}.

\begin{lem}
  \label{lem:arithprog}
  There exists an absolute constant $\lambda_0$ such that,
  for all $\lambda\ge\lambda_0$, and for all but at most
  $\lambda/\ln^2 \lambda$ primes $p$ in the range
  $(\lambda,2\lambda]$,
  there are at least $\lambda^{0.89}/\ln \lambda$ primes $q$
  in the range $(\lambda^{1.89},2\lambda^{1.89}]$
  such that $p|(q-1)$.
\end{lem}
\begin{proof}
  Set $K=0.53$, which means $(1.89)^{-1} < K < \tfrac{17}{32}$.
  Fixing $s=1$, and for any $R\ge 2$, Lemma~7 in \cite{Fou13}
  guarantees the existence of positive constants $\alpha_K$
  and $x_K$ such that the following holds: 
  For all $x > \max(x_K,R^{1/K})$,
  and for all but $R/\ln^2 R$ integers $r \in (R,2R]$,
  there are at least $\alpha_K x / (\varphi(r) \ln x)$
  primes $q$ in the range $(x,2x]$ such that $r|(q-1)$,
  where $\varphi(r)$ is the Euler totient function.

  Setting $\lambda_0 = \max(x_K, 3.78/\alpha_K)$,
  and letting $R=\lambda$, $r=p$, and $x=\lambda^{1.89}$,
  the statement of our lemma holds because
  $$\varphi(r)\ln x = (p-1)\ln \lambda^{1.89}
    < 3.78 \lambda \ln \lambda,$$
thus \quad\quad
  $\displaystyle\frac{\alpha_K x}{\varphi(r)\ln x}
  = \frac{\alpha_K \lambda^{1.89}}{(p-1)\ln \lambda^{1.89}}
  > \frac{\lambda^{0.89}}{\ln \lambda}.$
\end{proof}

\begin{procedure}[btp]
  \caption{GetPrimRoots($D,T,C,\mu$)\label{proc:getpqw}}
  \KwIn{$D \ge \deg F$; $T \ge \#F$; $C\ge \norm{F}_\infty$;
    $\mu\in(0,1)$; where $F\in\ZZ[x]$ is fixed but unspecified.}
  \KwOut{Prime $p$, primes $(q_1,\ldots,q_k)$, and integers
    $(\omega_1,\ldots,\omega_k)$; or \fail{}.}
  $m \gets \lceil \lg \tfrac{2}{\mu} \rceil$ \;
  $\lambda \gets \max\left(786, \lambda_0, \tfrac{20}{3\mu}mT(T-1)\ln D,
    1.35 \ln^{3.13} (2C)\right)$ \;
  $a \gets \lceil 1.1\ln(2C) \ln^2 \lambda \rceil$ \;
  \RepTimes{$m$}{
    $p \gets $\ref{proc:getprime}$(\lambda,\tfrac{\mu}{4m})$ \;
    $\A \gets a$ distinct even integers in
      $[2,2\lambda^{0.89}]$ \label{getpqw:chooseA}\;
    $Q, W \gets$ empty lists \;
    \ForEach{$a \in \A$}{
      $q \gets ap+1$ \;
      $\zeta \gets$ random nonzero element of $\ZZ_q$ \;
      \If{$q$ is prime and $\zeta^a\bmod q \ne 1$}{
        Add $q$ to $Q$ and $\omega=\zeta^a$ to $W$ \;
        \lIf{$\prod_{q\in Q} q \ge 2C$}{\Return $p$, $Q$, and $W$\label{getpqw:checkret}}
      }
    }
  }
  \Return \fail \;
\end{procedure}

Lemma~\ref{lem:arithprog} forms the basis for Algorithm~\ref{proc:getpqw}, where
we assume that the constant $\lambda_0$ is given. Since this constant
has not actually been computed, a reasonable strategy would be to choose
some small ``guess'' for $\lambda_0$ and run the algorithm until it
does not report failure. If the algorithm fails, it could be due to the
random prime $p$ being an ``exception'' in Lemma~\ref{lem:arithprog}, or
due to unlucky guesses for the primitive roots $\zeta$, or due to the
guessed constant $\lambda_0$ being too small. 
Because our primality tests are deterministic, failure due to
$\lambda_0$ being too small is detectable by the algorithm returning
\fail{}.

We state the running time and correctness, assuming $\lambda_0$ is 
sufficiently large, as follows.

\begin{lem}\label{lem:getpqwcost}
  Procedure \ref{proc:getpqw} has worst-case bit complexity
  \iftoggle{short}{%
  $\softohmu\left(\log C \cdot \polylog\left(T + \log
  D\right)\right).$ 
  }{%
  $$\softohmu\left(\log C \cdot \polylog\left(T + \log
  D\right)\right).$$
  }%
  With probability at least $1-\mu$, it returns a good \diffp{} $p$
  for $F$, primes $q_1,\ldots,q_k$ such that $\prod_i q_i \geq 2C$, and $p$th primitive roots
  modulo each $q_i$, $\omega_1,\ldots,\omega_k$.
\end{lem}
\begin{proof}
  The lower bound
  $\lambda \ge \max(786,1.35 \ln^{3.13} (2C))$ guarantees that there are
  sufficiently many even integers in the range $[2,2\lambda^{0.89}]$ in
  order for Step~\ref{getpqw:chooseA} to be valid, since for any
  $\lambda\ge786$, we have
  $\lambda^{0.89} > \lambda^{.32}\ln^2 \lambda > 1.1\ln (2C) \ln^2 \lambda$.

  For the running time, the outer loop does not affect the complexity in
  our notation because $m \in O(\log \tfrac{1}{\mu}) \in \softohmu(1)$.
  Observe also that
  $$\log \lambda \in \polylog\left(\lambda_0 + 
    T + \log D + \log C + \tfrac{1}{\mu}\right).$$
  The running time is dominated by the AKS primality tests in the inner
  loop, which are performed $O(m\log C \polylog(\lambda))$ times, each at
  cost $O(\polylog(\lambda))$, giving the stated worst-case 
  bit complexity.

  All of the checks for primality of $p$ and $q_i$'s, as well
  as the test that each $\omega_i$ is a $p$th primitive root of unity
  modulo $q_i$, are deterministic. Therefore the only possibility that
  the algorithm returns an incorrect result other than \fail{} is
  the probability that $p$ is not a good \diffp{} for $\supp(F)$.
  According to the proof of Lemma~\ref{lem:diffp},
  the condition $\lambda > \tfrac{20}{3\mu}mT(T-1)\ln D$, and using the union
  bound over all outer loop iterations, the probability that \emph{any}
  of the chosen $p$'s is not a good \diffp{} is less than $\mu/2$.

  Consider next a single iteration of the outer loop.
  This will produce a valid output unless insufficiently many good $q_i$'s
  and $\omega_i$'s are found for that choice of $p$.

  From Fact~\ref{fact:rosser} and Lemma~\ref{lem:arithprog}, the
  probability that $p$ is an exception to the lemma is at most
  $5/(3\ln\lambda)$, which is less than $\tfrac{1}{4}$ from the bound
  $\lambda\ge 786$.

  If $p$ is not an exception, then Lemma~\ref{lem:arithprog} tells us
  that the probability of each $q$ being prime is at least
  $\tfrac{1}{\ln \lambda}$. When $q$ is prime, since prime $p$ divides
  $(q-1)$, the probability that each $\zeta^a$ is a $p$-PRU in $\ZZ_q$
  is $(p-1)/p$, easily making the total probability of successfully adding to
  $Q$ and $W$ at each loop iteration at least $0.99/(\ln \lambda)$.

  Let $a \ge 1.1\ln(2C)\ln^2\lambda$ be the size of $\A$.
  By Hoeffding's inequality (\cite{Hoe63}, Thm.~1), the probability that
  fewer than $0.03 a / (\ln \lambda)$ integers are added to $Q$ after all
  iterations of the inner loop is at most
  $$\exp(-2 a (0.96/\ln\lambda)^2) < \exp(-2.02 \ln(2C)) < 0.25,$$
  where the last inequality holds because $C\ge 1$.

  Therefore, with probability at least $\tfrac{3}{4}$, and using
  again $\lambda \ge 786$, at least
  $$0.03 a/\ln \lambda = 0.033 \ln(2C) \ln \lambda > \ln(2C)/\ln
  \lambda$$
  integers are added to $Q$ each time through the inner loop.
  Since each $q_i>\lambda$, this means 
  $\prod_i q_i > 2C$, and the algorithm will return on
  Step~\ref{getpqw:checkret}.

  Combining with the probability that $p$ is an exception, we conclude
  that the probability the algorithm does \emph{not} return
  in each iteration of the outer loop is at most $1/2$. 
  As this is
  repeated $\lceil \lg \tfrac{2}{\mu}\rceil$ times, the probability is
  less than $\mu/2$ that the algorithm returns \fail{}. Using the union
  bound with the probability that any $p$ is not a good \diffp{}, we
  have the overall failure probability less than $\mu$.
\end{proof}

\section{Multiplying via Interpolation}\label{sec:basecase}

Let $F,G \in\ZZ[x]$ be sparse polynomials with 
$C=\norm{F}_\infty + \norm{G}_\infty$ and
$D=\deg F + \deg G$.
The subroutine \ref{proc:sumset} computes $\poss(F,G)$ 
by first reducing the degrees and
heights of the input polynomials and then multiplying them.
However, it cannot perform the multiplication using a recursive call to
\ref{proc:sparsemulzz} because the degrees are never reduced small
enough to allow the use of dense arithmetic in a base case.

Instead, we present here a ``base case'' algorithm
which, given $F$, $G$, and a bound
$S \ge \#F + \#G + \#(FG)$, computes $FG$, in time softly linear in $S,\log C$, and $\polylog(D)$.
Any algorithm with such running time suffices; we will use our own from \cite{ArnGieRoc14}, which is a Monte Carlo sparse
interpolation algorithm for univariate polynomials over
finite fields.

To adapt \cite{ArnGieRoc14} for multiplication over $\ZZ[x]$, 
we first choose a
``large prime'' $q > \max(2C, 2D)$ and treat $F,G$ and their product
$H=FG$ as polynomials over $\FF_q$. This size of $q$ ensures that no
extension fields are necessary.
The subroutine \ref{proc:bb} specifies how the unknown polynomial
$H=FG\in\FF_q[x]$
will be provided to the algorithm. It exactly matches the sorts of
black-box evaluations that \cite{ArnGieRoc14} requires. The entire
procedure is stated as \ref{proc:SM1}.

\begin{procedure}[btp]
  \caption{SparseInterpBB($F,G,\alpha,r$)\label{proc:bb}}
  \KwIn{$F,G\in\ZZ_q[x]$; $\alpha\in\ZZ_q$; $r\in\ZZ_{>0}$.}
  \KwOut{$H(\alpha z) \bmod (z^r-1)$, where $H=FG$.}
  $(\tilde{F}, \tilde{G}) \gets \left(F(\alpha z) \rem (z^r-1), G(\alpha z) \rem (z^r-1)\right)$\;
  $\tilde{H} \gets \tilde{F} \cdot \tilde{G} \rem (z^r-1)$ via 
  dense arithmetic\;
  \Return sparse representation of $\tilde{H}$ \;
\end{procedure}

\begin{lem}\label{lem:bb}
  The algorithm \ref{proc:bb} works correctly and uses
  $\softoh(S\log D\log q + r\log q)$ bit operations.
\end{lem}
\begin{proof}
  Correctness is clear. To compute $F(\alpha 
  z)$, we replace every term $cx^e$ of $F$ and $G$ with 
  $c\alpha^ex^e \in \ZZ_q$.  This costs 
  $\softoh(S\log D\log q)$ by binary powering.  Reducing 
  modulo 
  $(z^r-1)$ costs $\softoh( S(\log D+\log q))$ by Corollary 
  \ref{cor:modxp_m}.  Dense arithmetic costs $\softoh( r\log q)$ 
  bit 
  operations by Fact \ref{fact:dense_arithmetic}.  Summing these 
  costs yields $\softoh( S\log D\log q + r\log q)$.
\end{proof}

\begin{procedure}[Htbp]
  \caption{BasecaseMultiply($F,G,S,\mu$)\label{proc:SM1}}
  \KwIn{$F,G \in\ZZ[x]$; $S \ge \#F + \#G + \#(FG)$; 
    $\mu\in(0,1)$.}
  \KwOut{$H\in\ZZ[x]$ such that $\Pr[H \ne FG]<\mu$.}
  $q \gets $\ref{proc:getprime}$(2S\norm{F}_\infty\norm{G}_\infty +
    2\deg F + 2\deg G, \tfrac{\mu}{2})$ \;
  Call procedure MajorityVoteSparseInterpolate from \cite{ArnGieRoc14},
  with coefficient field $\ZZ_q$, sparsity bound $S$, degree bound 
  $\deg F + \deg G$, 
  error bound $\tfrac{\mu}{2}$, and black box 
  \ref{proc:bb}$(F,G,\cdot,\cdot)$. \;
\end{procedure}%

\begin{lem}
  The algorithm \ref{proc:SM1} correctly returns the product $H=FG$
  with probability at least $1-\mu$, and has bit complexity
  $\softohmu\left(
    S \log^2 D \left(\log C + \log D\right)
    \right).$
\end{lem}
\begin{proof}
  In order to use \ref{proc:bb} in the
  algorithm from \cite{ArnGieRoc14}, we simply replace the straight-line
  program evaluation on the first line of procedure ComputeImage with
  our procedure \ref{proc:bb}. Again, note that as the prime $q$ was
  chosen with $q>2D$, the MajorityVoteSparseInterpolate algorithm does
  not need to work over any extension fields.

  The correctness is guaranteed by the previous lemma,
  as well as Theorem 1.1 in \cite{ArnGieRoc14}. For the bit complexity,
  in Section 7 of \cite{ArnGieRoc14}, we see that the cost is dominated
  by $\softohmu(\log D)$ calls to the black box evaluation function,
  each of which is supplied $q\in \softoh(C+D)$, and 
  $r\in\softohmu(S\log D)$. Applying the bit complexity of Lemma \ref{lem:bb} gives the
  stated result.
\end{proof}

\section{Sumset Algorithm}\label{sec:sumset}

Let $\A, \B \in \ZZ \cap (-D,D)$ be nonempty, $R = \#\A+\#\B$, and
$S = \#(\A \setplus \B)$ throughout this section.
We prove as follows:\iftoggle{short}{\vspace{-1em}}{}
\begin{thm}\label{thm:supp}
Procedure $\ref{proc:sumset}(\A, \B, \mu)$ has bit complexity
$\softohmu(S\log D)$
and produces $\A \setplus \B$ 
with probability at least $1-\mu$.
\end{thm}
We compute the sumset $\A \setplus \B$ as $\supp(H)$, $H=FG \in 
 \ZZ[x^{-1},x]$, where
$
F = \sum_{a \in \A}x^a$ and $G = \sum_{b \in \B}x^b$.  Here $H$ has exponents in $(-2D, 2D)$ and 
$\norm{H}_{\infty} < R$.  Thus it suffices to construct the 
exponents of $H$ modulo $\ell \geq 4D$.  Moreover, we have 
$\supp(H)=\poss(F,G)$, and that
\begin{equation}\label{eqn:sumset_size}
R-1 \le \#(\A \setplus \B) = S \le R^2.
\end{equation}

\subsection{Estimating Sumset Output Size}
We first show how to compute a tighter upper bound on the true value of
$S=\#H=\#(A \setplus B)$.  To this end, let $p\in\bigoh(D)$ be a good \diffp{} for $\supp(H)$,
using the naive bound $R^2$ from \eqref{eqn:sumset_size}, and 
define the $H_1=F_1G_1$, where $F_1,G_1\in\ZZ[x]$ are defined by
$$
F_1 = F \rem (x^p-1), \qquad G_1 = G \rem (x^p-1).
$$
Then $\deg H_1<2p$ and each term $cx^e$ of $H$ corresponds to
either one or two terms in $H_1$, of
degrees $e \brem p$ and $(e \brem p)+p$. Therefore 
$$\#H_1/2 \leq \#H=S \leq \#H_1.$$

We will compute an approximation $S^* \approx S$ such that
$S^*/2 < \#H_1 \leq S^*$, and therefore $S^*/4 < S \leq S^*$.
To this end we present a test that, given $S^*$, always
accepts if $\#H_1 \leq S^*$ and probably rejects if $\#H_1 > 2S^*$.  We do this for $S^*$ initially $2$, doubling whenever the test rejects.

Given the current estimate $S^*$, we next choose a (1/2)-\diffp{} $q$ for 
the support of any $2S^*$-sparse polynomial with degree $2p>\deg 
H_1$,
and compute
$H^* = H_1 \bmod (x^q-1)$.  We work modulo
$m=R^2 > \lVert H \rVert_1 \geq \lVert H^* \rVert_\infty$, such that
none of the coefficients of $H^*$ vanish modulo $m$.  If $H_1$ is
$S^*$-sparse then $H^*$ is as well.  If $H_1$ has $2S^*$ terms
then, as fewer than $S^*$ terms of $H_1$ are in collisions, $H^*$ is
{\em not} $S^*$-sparse.  As no terms of $H_1$ vanish modulo $m$,
additional terms in $H_1$ can only increase $\#H^*$.

We choose $q$ so that the test is correct with probability at least
$3/4$. By iterating $\lceil 8\ln(8/\mu)\rceil$ times, by Hoeffding's
inequality, the probability is at least $1-\mu/4$ that the test 
runs correctly in at least half of the iterations. As $\#H_1 < 
2R^2$, it suffices that the test is correct $\lceil \log_2 R + 1 
\rceil$ times.

The \ref{proc:sumset} procedure performs this test on lines
\ref{proc:ss:size0}--\ref{proc:ss:sizen}.
By Corollary \ref{cor:modxp_m} the respective costs of 
constructing $F^*$ and $F^* \bmod (R^2, x^q-1)$ are 
$\softohmu(R\log D)$ and $\softohmu( R\log p\log R)$, and similarly 
for $G^*$.  The cost of the dense arithmetic here is $\softohmu( 
q\log R)$.  Given that $p \in \softohmu(D), q \in \softohmu( 
S\log p)$, the total bit-cost of this part is
$\softohmu(S\log D)$.

\subsection{Computing Sumset}

Armed with the bound $S^*/4 < S \le S^*$, 
we aim to compute $H=FG$. 
Our approach is to compute images 
\begin{align*}
H_1 &= F_1G_1 
\mod(\ell^2, x^p-1), \\
H_2 &= F((\ell+1)x)G((\ell+1)x) \mod (\ell^2, x^p-1),
\end{align*}
for an integer $\ell = 8D \ge \max(\deg H, \norm{H}_1)$.

Since the coefficients of $H_2$ are scaled by powers of $(\ell+1)$, a
single term $c x^e$ in the original polynomial $H$ becomes
$c x^{e\brem p}$ in $H_1$ and $c (\ell+1)^e x^{e\brem p}$
in $H_2$, and if they are uncollided we can discover 
$(\ell+1)^e$ by computing their quotient. Modulo $\ell^2$, this
quotient $(\ell+1)^e$ is simply $e\ell+1$, from which we can 
obtain the exponent $e$.
This idea is similar to the ``coefficient ratios'' technique suggested
by \cite{HL15}, but working modulo $\ell^2$ allows us to avoid costly
discrete logarithms.  Procedure \ref{proc:sumset} contains the complete description.

\begin{procedure}[htb]
  \caption{Sumset($\A,\B,\mu$)\label{proc:sumset}}
  \KwIn{$\A,\B \subseteq \ZZ$ with $\#\A + \#\B = R$ and $\max_{k \in A \cup B}|k| < D$; $\mu \in 
  (0,1)$.}
  \KwOut{Set $\S\subset\ZZ$ such that
  $\Pr[\S\ne \A\setplus\B]<\mu$.}
  $p \leftarrow $\ref{proc:diffp}$(R^2, 4D, 1, \mu/4)$\;
  $(F_1, G_1) \gets \left(\sum_{a \in \A}x^{a \brem p}, \sum_{b \in 
    \B}x^{b \brem p}\right)$\;
  \medskip

  $S^* \leftarrow 2$ \label{proc:ss:size0} \;
  \RepTimes{$\lceil\max(8\ln(8/\mu), \log_2 R+1)\rceil$}{
    $q \gets $\ref{proc:diffp}$(2S^*,2p,\frac{1}{2},\frac{3}{4})$ \;
    $H^* \leftarrow F_1 G_1 \bmod (R^2, x^q-1)$, via dense arithmetic \;
    \lIf{$\#H^*>S^*$}{$S^* \gets 2S^*$ \label{proc:ss:sizen}}
  }
  \medskip

  $\ell \gets 8D$ \;
  $(F_2, G_2) \gets \left(\sum_{a \in \A}(a\ell+1)x^{a \brem p}, 
  \sum_{b \in \B}(b\ell+1)x^b\right)$\;
  $H_1 \leftarrow $\ref{proc:SM1}$(F_1, G_1, S^*, \mu/4)$\;
  $H_2 \leftarrow $\ref{proc:SM1}$(F_2, G_2, S^*, \mu/4)$\;
  \lFor{$j=1,2$}{$H_j \leftarrow H_j \bmod (\ell^2, x^p-1)$}
  \medskip

  $\S \leftarrow$ empty list of integers \;
  \For{every nonzero term $cx^e$ of $H_1$}{
    $c' \leftarrow$ coefficient of degree-$e$ term of $H_2$\;
    \If{$c\mid c'$ and $\ell \mid (c'/c - 1)$ as integers}{
      Add $(c'/c - 1)/\ell$ to $\S$ \;
    }
    \lElse(\hfill //cannot reconstruct an exponent){\Return $\fail$}
  }
  \Return $\S$ \;
\end{procedure}

\ref{proc:sumset} has four steps that are probabilistic: 
choosing a good \diffp{} $p$, estimating the sumset size
$S = \#(\A \setplus \B)$,
and constructing $H_1$ and $H_2$.
As each is set to fail with probability less than $\mu/4$, 
\ref{proc:sumset} succeeds with probability at least $1-\mu$.

We now analyze the total cost of this algorithm.  
\ref{proc:diffp} produces $p$ of size
$\log p < \polylog(R + \log D + \tfrac{1}{\mu})$.
Constructing $F_1, F_2, G_1, G_2$ at the beginning,
and the reduction of $H_1,H_2$ modulo 
$(\ell^2,x^p-1)$ at the end, both 
cost $\softohmu(S\log D)$ bit operations.

The search for $S^*$ costs $\softohmu(S\log D)$
from the previous section.  Finally, as 
$\norm{F_1}_\infty < \ell/2$, $\deg F_1 < p$,
and similarly for $F_2$, $G_1$, and $G_2$,
the sparse multiplications due to \ref{proc:SM1} also costs 
$\softohmu(S \log D)$ bit operations.
These dominate the complexity as stated in
Theorem~\ref{thm:supp}.

\section{Multiplication with support}\label{sec:knownsupp}

We turn now to the problem of multiplying sparse $F,G \in \ZZ[x]$,
provided some $\S \supseteq \supp(FG)$.  This algorithm is used twice in our overall 
multiplication algorithm: first with large $\S = \poss(F,G)$ 
but small coefficients, then with the actual support 
$\S = \supp(FG)$ but full-size coefficients.

\iftoggle{short}{}{
The bit-complexity of our algorithm is summarized in the following
theorem.
}

\begin{thm}\label{thm:knownsupp}
Given $F, G \in \ZZ[x]$ and a set $\S \subset \ZZ_{\ge 0}$
such that $\supp(FG)\subseteq \S$, the product $FG$ can be computed in time
$\softohmu\left(\left(\#F + \#G + \#\S\right)
  \left(\log C + \log D\right)
\right),$
where $C = \norm{F}_\infty + \norm{G}_\infty$ and $D > \deg F + \deg G$.
\end{thm}

This is $\softohmu$-optimal, as it
matches the bit-size of the inputs.

Our algorithm requires a small randomly-selected good \diffp{} $p$ 
with $\bigoh(\log S + \log\log D)$ bits,
and a series of pairs $(q,\omega)$, where each $q$ is
a slightly larger prime with $\bigoh(\log p)$ bits,
and $\omega$ is an order-$p$ element in $\ZZ_q$.
These numbers are provided by \ref{proc:getpqw} (Sec.~\ref{sec:redexp}).
Our algorithm works by first reducing the exponents modulo $p$, then
repeatedly reducing the
coefficients modulo $q$ and performing evaluation and interpolation at
powers of $\omega$. This inner loop follows exactly the algorithm of
\cite{vdHLec13} and \cite{KY89} for applying a transposed Vandermonde
matrix and its inverse. Since $p$ is a good \diffp{} for the support of
the product, there are no collisions and this gives us each
coefficient modulo $q$. The process is then repeated $\bigoh(\log C)$ times
in order to recover the full coefficients via Chinese remaindering.

\subsection{Comparison to prior work}

Without affecting the complexity, we may assume that 
$\S$ contains the support of the inputs too, i.e.,
$\supp(F)$ and $\supp(G)$. We also assume
that $\max \S = \deg FG$, such that no $e \in \S$
is too large to be an exponent of $FG$.
Under these assumptions, and writing $S = \#\S$,
the stated complexity of our algorithm is simply
$\softohmu(S(\log C + \log D))$.

The problem of computing the coefficients of a sparse product, 
once the exponents of the product are given, has been recently and 
extensively investigated by van der Hoeven and Lecerf, where they 
present an algorithm whose bit complexity (in our notation) is
$$\textstyle\softohmu\left(\left(\sum_{e \in \S}\log e\right) (\log 
D + \log C)\right)$$
(\cite{vdHLec13}, Corollary 5). 
As $\sum_{e\in\S}\log e \in \bigoh(S\log D)$, 
the algorithm here saves a factor
of at most $\bigoh(\log D)$ in comparison, which could be 
substantial if the exponents are very large.

Their algorithm is more efficient if the support superset $\S$ is \emph{fixed}, in
which case they can move the most expensive parts into precomputation
 and compute the result in the same soft-oh time as our approach,
$\softohmu(S (\log D + \log C))$. Furthermore, the support bit-length
$\sum_{e\in\S}\log e$ is at most $\bigoh(S\log D)$, but could be as small as
$\Omega(S\log S + \log D)$, for example if the support contains only a
single large exponent. In such cases our savings is only on the order of
$(\log D)/S$.

\subsection{Transposed Vandermonde systems}

Applying transposed Vandermonde systems, and their inverses, is an
important subroutine in sparse interpolation algorithms, and
efficient algorithms are discussed in detail by \cite{KY89} and
\cite{vdHLec13}. We restate the general idea here and refer the reader
to those papers for more details.

If $\tilde{F}$ is a dense polynomial, it is well known that applying the
Vandermonde matrix $\V(\theta_1,\ldots,\theta_D)$ to a vector of
coefficients from $\tilde{F}$ corresponds to evaluating $\tilde{F}$
at the points $\theta_1\ldots,\theta_D$.
Applying the inverse Vandermonde matrix corresponds to interpolating
$\tilde{F}$ from its evaluations at those points.
The product tree
method can perform both of these using 
$\softoh(D)$ field operations (\cite{MCA}, Chapter 10).

If $F$ is instead a sparse polynomial
$F = \sum_{e\in\S} c_1 x^e$, evaluating $F$ at consecutive powers of
a single high-order element $\omega$ corresponds to multiplication 
with the transposed Vandermonde matrix:
$$
  \V(\omega^{e_1},\ldots,\omega^{e_S})^T
  (c_1,\ldots,c_S)^T
  = (F(1),\ldots,F(\omega^{S-1}))^T
$$

The transposition principle tells us it is possible to
compute the maps $V^T$ and $(V^T)^{-1}$ in essentially the same time as
dense evaluation and interpolation. %
In particular, if the coefficients $c_i$ are in the modular ring
$\ZZ_q$, then the transposed Vandermonde map and its inverse can be
computed using $\softoh(S\log q)$ bit operations \cite{vdHLec13}.

\subsection{Statement and analysis of the algorithm}

\begin{procedure}[htbp]
  \caption{SparseMulCoeffs($\S, F, G, \mu$)\label{proc:knownsupp}}
  \KwIn{Exponents $\S = (e_1, e_2,\ldots e_S)$;
    coefficient lists $(f_1,\ldots, f_S)$ and
    $(g_1,\ldots, g_S)$, with $F,G\in\ZZ[x]$
    implicitly defined as 
    $F = \sum_{1\le i \le S} f_i x^{e_i}$ and
    $G = \sum_{1\le i \le S} g_i x^{e_i}$; 
    error bound $\mu\in(0,1)$.}
  \KwOut{$(h_1,\ldots,h_S)\in\ZZ^S$ such that, with probability
    least $1-\mu$, $FG=\sum_{1\le i\le S} h_i x^{e_i}$.}
  $C \gets (\max_{1\le i\le S} |f_i|)(\max_{1\le i \le S} |g_i|) S$ \;
  $p, Q, W \gets $\ref{proc:getpqw}$(\max\S, \#\S, C, \mu)$ 
    \label{proc:ks:choosep}\;
  $\S_p \gets (e_1 \brem p, e_2 \brem p,\ldots, e_S \brem p)$
    \label{proc:ks:reduce}\;
  $H \gets$ list of $S$ empty lists of integers \;
  \ForEach{$(q,\omega) \in Q,W$\label{proc:ks:eachq}}{
    \ForEach{$e_{ip} \in \S_p$}{
      $v_i \gets \omega^{e_{ip}} \in \ZZ_q$ by binary powering 
        \label{proc:ks:binpow}\;
    }
    $\mathbf{a} \gets \V(v_1,\ldots,v_S)^T (f_1,\ldots,f_S)^T \in \ZZ_q^S$ 
      \label{proc:ks:vtrans1}\;
    $\mathbf{b} \gets \V(v_1,\ldots,v_S)^T (g_1,\ldots,g_S)^T \in
    \ZZ_q^S$
      \label{proc:ks:vtrans2}\;
    $\mathbf{c} \gets (a_1 b_1,\ldots, a_S b_S)^T \in \ZZ_q^S$ 
      \label{proc:ks:pairmul}\;
    \If{$\V(v_1,\ldots,v_S)$ is invertible}{
      $(h_{1p}, \ldots, h_{Sp}) \gets (\V(v_1,\ldots,v_S)^T)^{-1}
      \mathbf{c} \in \ZZ_q^S$
        \label{proc:ks:vtinv}\;
      \lFor{$1 \le i \le S$}{Add $h_{ip}$ to the list $H[i]$}
    }
  }
  \For{$1 \le i \le S$}{
    $h_i \gets$ Chinese remaindering from images $H[i]$ modulo integers
    in $Q$ \;
  }
  \Return $(h_1,\ldots, h_S)$ \;
\end{procedure}

\begin{lem}
  Procedure~\ref{proc:knownsupp} works as stated when $\S \supseteq \supp(FG)$.  In any case it has bit-complexity
  $$\softohmu\left(
      \shortsum_{e\in\S}\log e + S\log C
  \right).$$
\end{lem}

\begin{proof}
We first analyze the probability of failure when $\S \supseteq \supp(FG)$.
The randomization is in the choices of
$p$, $q$, and $\omega$; problems can occur if these lack the required
properties.

If $p$ is a good \diffp{} for $\supp(FG)$, then by definition there will be no collisions
in $\S_p$. Furthermore, if $\omega$ is a $p$th root of
unity modulo $q$, then there are no collisions among the values
$(v_1,\ldots,v_S)$, so $\V(v_1,\ldots,v_S)$ is invertible modulo $q$.
Algorithm \ref{proc:getpqw} ensures this is the case with high
probability, and if so the algorithm here faithfully computes each
coefficient $h_i$ modulo $q$.

Conversely, if there are no collisions in $\S_p$,
and if no zero divisors modulo $q$ are encountered in the application of
the Vandermonde matrix and its inverse, then the algorithm correctly
computes then values $h_i \bmod q$, even if $p$ is not prime or some $\omega \in \ZZ_q$ is not actually a
$p$th root of unity.

Therefore all failures in choosing tuples $p,q,\omega$ are either
detected by the algorithm or do not affect its correctness. Since that
is the only randomized step, we conclude that the entire algorithm is
correct whenever the input exponent set $\S$ contains the support of the
product.

For the complexity analysis first
define $D = \deg (FG)$.
Step~\ref{proc:ks:choosep} costs 
$\softohmu(\log C \cdot \polylog(S + \log D))$ bit operations
by Lemma~\ref{lem:getpqwcost}.
Reducing each exponent $e_i$ modulo $p$, on step~\ref{proc:ks:reduce},
can be done for a total of
$\softohmu(\shortsum_{e\in\S} \log e)$ bit operations.

Now we examine the cost of the for loop that begins on
step~\ref{proc:ks:eachq}.
As the exponents are now all less than $p$, computing each $v_i$ on
step~\ref{proc:ks:binpow} requires only $\bigoh(\log p)$ operations modulo
$q$, for a total of $\bigoh(S\log p \log q)$, which is 
$\softohmu(S \cdot \polylog(\log C + \log D))$ bit operations. From before we
know that
applying the transposed Vandermonde matrix and its
inverse takes $\softohmu(S\log q)$, or 
$\softohmu(S \cdot \polylog(\log C + \log D))$ bit operations.

Because $\#Q \le \lceil\log_p (2C)\rceil$,
the loop on step~\ref{proc:ks:eachq} repeats $\bigoh(\log C)$ times, for a
total cost of $\softohmu(S \log C \cdot \polylog(\log D))$ bit operations.
This also bounds the cost of the Chinese remaindering in the final loop.
\end{proof}

\section{Multiplication algorithms}\label{sec:smul}

The complete multiplication algorithm over $\ZZ[x]$ that
was outlined in the introduction is presented as
\ref{proc:sparsemulzz}.

\begin{procedure}[Htbp]
  \caption{SparseMultZZ($F,G$)\label{proc:sparsemulzz}}
  \KwIn{Sparse $F,G\in\ZZ[x]$; $\mu \in (0,1)$.}
  \KwOut{Sparse $H\in\ZZ[x]$,
    such that $\Pr[H\ne FG] < \mu$.}
  $\S_1 \gets $\ref{proc:sumset}$(\supp(F), \supp(G), \tfrac{\mu}{4})$ 
    \label{proc:smz:sumset}\;
  $C_H \gets \norm{F}_\infty \norm{G}_\infty \max(\#F,\#G)$ 
    \label{proc:smz:height} \;
  $p \gets $\ref{proc:vanp}$(\#\S_1, C, 1, \tfrac{\mu}{4})$ \;
  $H_1 \gets $\ref{proc:knownsupp}$(F \rem p, G \rem p, \S_1,
    \tfrac{\mu}{4})$ \label{proc:smz:ks1} \;
  $\S_2 \gets \supp(H_1 \rem p)$ \;
  \Return \ref{proc:knownsupp}$(F, G, \S_2, \tfrac{\mu}{4}), \S_2$ \;
\end{procedure}

\begin{proof}[\iftoggle{short}{}{Proof }of Thm.~\ref{thm:main}]
  Unless failure occurs, we have $\S_1=\poss(F,G)$. 
  Every coefficient in $H$ is a sum of products of coefficients in $F$
  and $G$, so the value $C_H$ computed on step \ref{proc:smz:height} is
  an upper bound on $\norm{H}_\infty$, and $p$ is a good \vanp{} for
  $H$. Thus $\S_2 = \supp(H \rem p) = \supp(H)$, so the final step
  correctly computes the product $FG$.

  By the union bound, Lemma
  \ref{lem:vanp} and Theorems \ref{thm:supp} and 
  \ref{thm:knownsupp},
  the probability of failure is less than $\mu$. 

  Writing $T=\#\poss(F,G)$ and $S=\#\supp(FG)$, we see that
  $\log p \le \polylog(T + \log C + \tfrac{1}{\mu})$, thus 
  steps \ref{proc:smz:sumset} and \ref{proc:smz:ks1} 
  contribute
  $\softohmu(T\log D)$ to the overall cost, whereas the last step
  costs $\softohmu(S(\log D + \log C))$ bit operations. As
  $S\le T$, the total bit complexity is $\softohmu(T\log D + S\log C)$,
  as required.
\end{proof}

We now consider extensions of this algorithm to
positive and negative exponents (Laurent polynomials), multiple variables,
and other common coefficient rings, using Kronecker substitution.  This
is stated in the following theorem.

\begin{restatable}{thm}{extensions}
  Let $F,G \in \R[x_1^{\pm 1},\ldots,x_n^{\pm 1}]$
  be sparse Laurent polynomials over $\R$, where $\R=\ZZ$,
  $\ZZ_m$, $\mathsf{GF(q^e)}$,
  or $\QQ$. The product $FG$ can be computed using
  \iftoggle{short}{$\softohmu(T(n\log D + B))$ }%
  {$$\softohmu\left(T\left(n\log D + B\right)\right)$$}
  bit operations, where
  $T = \#\poss(F,G)$ is the structural sparsity of the product,
  $D>\max_i |\deg_i (FG)|$, and $B$ is the largest bit-length of any
  coefficient in the input or output.
\end{restatable}

\begin{proof}
Write the output polynomial $H=FG$ as
\[H = \textstyle\sum_{i=1}^T c_i x_1^{e_{i1}} x_2^{e_{i2}} \cdots x_n^{e_{in}},\]
where each $c_i\in\R$ and each $e_{ij}$ satisfies
$|e_{ij}| < D$. %

We first apply the Kronecker substitution, %
providing an easily-invertible map between
$\R[x_1^{\pm 1},\ldots,x_n^{\pm 1}]$ and $\R[z^{\pm 1}]$: $x_i \mapsto z^{D^{i-1}}$
for $1 \leq i \leq n$.
This increases the degree to $D^n$, such that the
logarithm of this degree $O(n\log D)$,
matching the exponent bit-size in the multivariate representation. 

The algorithm \ref{proc:sumset} already handles negative exponents (i.e.,
Laurent polynomials) explicitly. The other primary subroutine to 
\iftoggle{short}{}{procedure} \ref{proc:sparsemulzz} is
\ref{proc:knownsupp}, which only uses the exponents in the set $\S_p$,
which are reduced modulo $p$ and therefore cause no difficulty if they
are negative. Thus the multiplication algorithms handle univariate 
Laurent polynomials without any changes.

To extend the multiplication algorithm and it subroutines beyond
$\R=\ZZ$, we use that our algorithm is also
softly-linear in the input \emph{heights}. This allows us to
adapt to many coefficient domain that provides a natural mapping to the integers,
and to preserve softly-linear time if that mapping
provides only a softly-linear increase in size.

For a modular ring $\R=\ZZ_m$, we can trivially treat the inputs as actual integers,
then reduce modulo $m$ after multiplying.
For a finite field $\R=\mathsf{GF}(p^d)$, elements are typically
represented as polynomials over $\ZZ_p[z]$ modulo a degree-$d$ irreducible
polynomial, so these coefficients can be converted to integers using a
low-degree Kronecker substitution.
For the rationals $\R=\QQ$, we might choose a prime $q$ larger than
the product of the largest numerator and denominator in the output,
multiply modulo $q$, then use rational reconstruction to recover the
actual coefficients.

In all the above cases, there is growth in the bit-length of
coefficients, but only in poly-logarithmic terms of input and 
output size, therefore not affecting the soft-oh complexity. 
The only downside is that we are no longer able to split the cost neatly
between $T=\poss(F,G)$ and $S=\supp(FG)$ because the unreduced integer
polynomial product might have nonzero coefficients which are really
zeros in \R.
\end{proof}

\section*{Acknowledgements}
\addcontentsline{toc}{section}{Acknowledgements}
We thank Mark Giesbrecht for helpful discussions on the development of
this paper (and much more).

We thank Timothy Chan for bringing \cite{CH02} to our attention,
and the ISSAC 2015 referees for their constructive feedback.

The first author is supported by the Natural Sciences and Engineering Research Council of Canada (NSERC).

The second author is supported by National Science Foundation
award no.\ 1319994, ``AF: Small: RUI: Faster Arithmetic for Sparse
Polynomials and Integers.''

\renewcommand{\bibpreamble}{\addcontentsline{toc}{section}{References}}
\bibliographystyle{plainnat}

\newcommand{\Gathen}{\relax}\newcommand{\Hoeven}{\relax}

\end{document}